\documentclass[letterpaper]{article} 
\usepackage{aaai2026}  
\usepackage{times}  
\usepackage{helvet}  
\usepackage{courier}  
\usepackage[hyphens]{url}  
\usepackage{graphicx} 
\usepackage{natbib}  
\usepackage{caption} 
\usepackage{tabularx}
\usepackage{algorithm}
\usepackage{algorithmic}
\usepackage{amsmath}
\usepackage{amsthm}
\usepackage{amssymb}
\usepackage{bm}
\usepackage{array}
\usepackage{algorithm}
\usepackage{algorithmic}
\usepackage{threeparttable}
\usepackage{array}
\usepackage{booktabs}
\usepackage{subfloat}
\usepackage{newfloat}
\usepackage{listings}
\usepackage{multirow} 
\usepackage{newfloat}
\usepackage{listings}
\usepackage{subfigure}
\DeclareCaptionStyle{ruled}{labelfont=normalfont,labelsep=colon,strut=off} 
\floatstyle{ruled}
\newfloat{listing}{tb}{lst}{}
\floatname{listing}{Listing}
\newtheorem{theorem}{Theorem}

\title{Achieving Equilibrium under Utility Heterogeneity: An Agent-Attention Framework for Multi-Agent Multi-Objective Reinforcement Learning}
\author{
Zhuhui~Li\textsuperscript{\rm 1},
Chunbo~Luo\textsuperscript{\rm 1},
Liming~Huang\textsuperscript{\rm 2},
Luyu~Qi\textsuperscript{\rm 3},
and~Geyong~Min\textsuperscript{\rm 1}
}
\affiliations{
    \textsuperscript{\rm 1}Department of Computer Science, University of Exeter, EX4 4QF, UK\\ zl462@exeter.ac.uk,
    c.luo@exeter.ac.uk , 
    g.min@exeter.ac.uk\\
    \textsuperscript{\rm 2}School of Mechanical and Electrical Engineering, Central South University, Hunan, China\\huanglm.me@gmail.com\\
    \textsuperscript{\rm 3}Faculty of Science and Engineering, University of Bristol, BS8 1UB Bristol, U.K\\luyu.qi@bristol.ac.uk

}

\usepackage{bibentry}
\begin{document}

\maketitle

\begin{abstract}
Multi-agent multi-objective systems (MAMOS) have emerged as powerful frameworks for modelling complex decision-making problems across various real-world domains, such as robotic exploration, autonomous traffic management, and sensor network optimisation. MAMOS offer enhanced scalability and robustness through decentralised control and more accurately reflect inherent trade-offs between conflicting objectives. In MAMOS, each agent uses utility functions that map return vectors to scalar values. Existing MAMOS optimisation methods face challenges in handling heterogeneous objective and utility function settings, where training non-stationarity is intensified due to private utility functions and the associated policies. 
In this paper, we first theoretically prove that direct access to, or structured modeling of, global utility functions is necessary for the Bayesian Nash Equilibrium under decentralised execution constraints. To access the global utility functions while preserving the decentralised execution, we propose an Agent-Attention Multi-Agent Multi-Objective Reinforcement Learning (AA-MAMORL) framework. Our approach implicitly learns a joint belief over other agents’ utility functions and their associated policies during centralised training, effectively mapping global states and utilities to each agent's policy. In execution, each agent independently selects actions based on local observations and its private utility function to approximate a BNE, without relying on inter-agent communication.
We conduct comprehensive experiments in both a custom-designed MAMO Particle environment and the standard MOMALand benchmark. The results demonstrate that the accessibility to global preferences and our proposed AA-MAMORL significantly improves performance and consistently outperforms state-of-the-art methods.

\end{abstract}

\section{Introduction}
Multi-Agent Multi-Objective Systems (MAMOSs) have been spotlighted in real-world applications, such as balancing exploration and exploitation in networked robotic systems \cite{paine2024model}, managing the trade-off between efficiency and energy consumption in autonomous traffic control \cite{shi2021multi}, and optimising the trade-off between resolution and coverage in mobile sensor monitoring tasks \cite{hayat2020multi}. In contrast to single-agent systems where the decision burden and failure risk are centralised, the MAMOS distributes both computation and control across agents. This decentralisation enhances system scalability and enables resilience and robustness to partial agent failures \cite{he2021secure}. Meanwhile, the multi-objective formulation in MAMOSs also better reflects the inherent trade-offs in real-world applications. To be specific, most real-world systems require trade-offing between multiple, often conflicting, performance metrics. Representative MO settings include energy efficiency \cite{niu2023active}, energy performance index \cite{chang2023does}, and water use efficiency \cite{mallareddy2023maximizing}, often in conjunction with advanced integrated technological paradigms such as Simultaneous Wireless Information and Power Transfer \cite{wei2021resource}, Integrated Sensing and Communication \cite{qi2022integrating}, and piezoelectric roads \cite{jiang2023research}.

Although steady progress has been made in the development of MAMOSs, several critical challenges remain, including the joint interdependencies among objectives and agents, the dynamic nature of real-world systems, and the non-differentiability of many environments \cite{wong2023deep}. The heterogeneity and diversity of reward (corresponding to objective) and utility function settings in MAMOS further pose significant challenges to the MAMO optimisation. As discussed in \cite{ruadulescu2020multi}, the utility function is defined as a mapping from the vectorised rewards to a scalar utility. The utility function in this paper is referred to as the preference, where the weighted sum based on rewards and preferences forms the most basic utility function for all agents. The decision-making problems in MAMOS can be categorised into five settings, based on the combinations of reward and utility function types.
 On the reward side, agents receive either a team reward, where all agents share the same reward vector reflecting collective performance, or an individual reward, where each agent obtains a personalised reward vector. On the utility side, agents optimise a shared team utility, pursue a social choice utility that aggregates all agents’ rewards into a global social welfare function, or optimise their own individual utility, in which each agent maintains a private function. Examples of these combinations in real-world companies are illustrated in Fig.~\ref{fig1}.
\begin{figure}[!t] 
	\centering
	\includegraphics[width=\columnwidth]{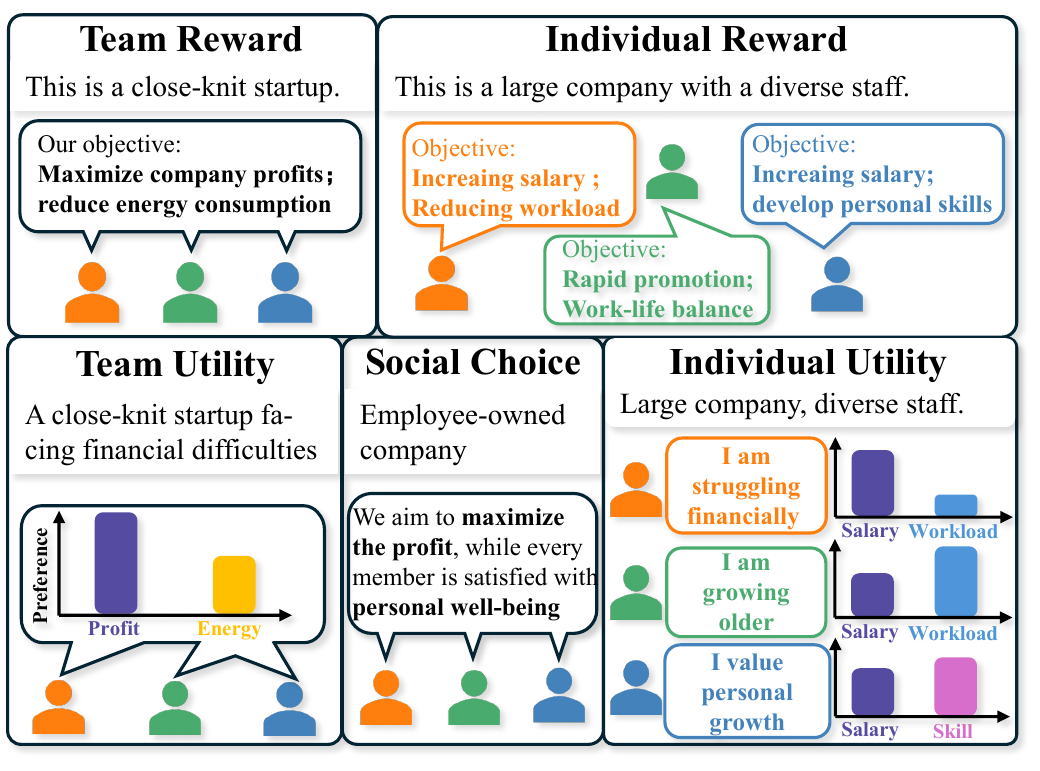}
	\caption{Illustrative examples of the five types of reward and utility function setting, using different companies as analogies.}
	\label{fig1}
\end{figure}

In the simplest setting: team utility with team reward, the problem can be simplified into a single-agent formulation, where one entity optimises its joint policy by acting over the entire joint action space \cite{ruadulescu2020multi}.
In the more challenging team utility with individual reward setting, Hu et al. \cite{hu2023mo} have proposed MO-MIX to solve Partially Observable MO Markov Decision Process (POMOMDP)s within the Centralised Training with Decentralised Execution (CTDE) framework. Their approach incorporates a Multi-Objective Conditioned Agent Network for evaluating action-values under the team utility, a parallel mixing network for estimating joint action-values, and a preference-based exploration strategy to promote diverse and well-distributed Pareto-optimal policies.

However, the optimisation in the most general setting: individual utility functions, remains unsolved. But this setting is critical as it is able to reflect most real-world scenarios where agents act based on personal intentions, constraints, or their distinct roles in the environment. Designing a general optimisation framework for this setting is challenging, with the reason that each agent must optimise its policy according to its own utility function, while the reward of that policy depends on the joint policy with other agents under their individual utility functions. Since these utility functions are heterogeneous and potentially conflicting, each agent’s policy cannot be updated and performed in isolation \cite{assos2024maximizing}. 

This naturally leads to the optimisation in MAMOSs as a Bayesian game, where each agent possesses private information, such as utility functions, objectives, or local observations, and must form beliefs about the policies of others to make its own optimal decision. The canonical solution concept in such settings is the Bayesian Nash Equilibrium (BNE) \cite{saglam2025bayesian}, a joint policy in which no agent can unilaterally improve its expected utility, given its beliefs about the types and policies of others. Achieving BNEs in MAMOS settings is non-trivial due to the intensified non-stationarity introduced by other agents’ private utility functions, multiple rewards settings, and the associated policies \cite{assos2024maximizing}.

While CTDE has become a dominant paradigm in MARL to optimise the individual policy (actor) by forming beliefs about the policies of others in the centralised critic network \cite{li2025gtde}, it does not guarantee convergence to BNE when agents’ utility functions are partially or entirely unknown to each other in MAMOSs, since the CTDE framework will still fail to capture or model utility-based policies with the unknown utility function. 
As a result, the learning process is still highly non-stationary from each agent’s perspective, and the BNE remains attainable.
Thus, BNE in MAMOs requires that every agent possesses a model of others’ utility functions and the corresponding policies those utilities induce. This fundamental requirement constitutes the main problem to be addressed in this paper:
\textbf{How to map the global state, utility functions, and the associated joint policies to each agent’s individual policy, so that each agent can learn its own decentralised optimal policy that maximise the global utility, even under heterogeneous objective and utility function settings.}

We find that belief modeling on the joint policy becomes tractable, and BNE becomes theoretically attainable when each agent’s utility function is a deterministic function of the state or observation (e.g., $\bm{w}_i = g(o_i)$). This case also better aligns with many real-world applications. For instance, when purchasing train tickets, a traveler facing tight time constraints may prioritise on-time arrival over cost, whereas one planning in advance during a pre-sale period is more likely to prefer the lowest possible price. Building on this finding, we first prove that when each agent’s utility function is either directly observable or can be modeled from local observations, convergence to a BNE becomes theoretically attainable.
Then, we propose an agent-attention MAMORL framework. This framework implicitly learns a joint belief over other agents’ utility functions and their corresponding policies through centralised agent-attention-based critic training. Thus, each agent learns a mapping from the global state and utility functions to its own policy under such belief. The learned distributed policy of each agent can be executed using only its local observations and private utility function. And no agent has an incentive to unilaterally deviate from its decision given the system-wide context. Thereby the BNE of MAMOSs can be approximated in a fully decentralised and communication-free setting. Our main contributions are summarised as follows:

1.	We formalise the MAMOS optimisation as a general POMOMDP, capable of abstracting various applications with heterogeneous reward and utility function settings. 

2. We bridge between the POMOMDP and Bayesian games, and rigorously prove that even under the CTDE paradigm, the utility function is essential for achieving BNE in decentralised decision-making.

3.	We propose an agent-attention MAMORL framework for scenarios where utility functions are deterministic functions of agents’ local observations. This framework enables the optimal decentralised policy conditioned solely on its local observation and private utility function for each agent.

4.	We conduct comprehensive experiments in the MAMO Particle environment and the MOMALand benchmark. The results demonstrate that the global utility functions and our proposed AA-MAMORL framework consistently improve multiple MO metrics.

\section{Preliminaries}
\subsection{Partially Observable MO Markov Decision Process}
POMOMDP is defined by 
the tuple:
\(
  \Omega = \left(S, \bm{\mathcal{A}}, \bm{R}, \bm{W}, \bm{P}_o,\bm{P}_{wt}, P_t, P_0, \gamma\right)
\). Within this process, $S$ is the state space describing the possible states of all agents and the environment, \(\mathcal{A}_1, \ldots, \mathcal{A}_N \in \bm{\mathcal{A}}\) and \(\bm{w}_1, \ldots, \bm{w}_N \in \bm{W}\) are the action and preference spaces for all agents.  At each time slot, agent \(i\) first uses its observation function \(P_o^i \in \bm{P}_o: S \rightarrow O_i\) to obtain its own observation \(o_i\) based on the state \(s\sim S\). Each agent $i$ uses its MO policy \(\pi_i: O_i \times \bm{w}_i \rightarrow \mathcal{A}_i\) to select its action $a_i\sim \mathcal{A}_i$ since observations coupled with the utility function affect agents’ decisions jointly. The transition function \(P_t: S \times \mathcal{A}_1 \times \ldots \times \mathcal{A}_N \rightarrow S\) transits the current state \(s_t\) to \(s_{t+1}\).
The preference transition function $P_{wt}^t \in \mathbf{P}_{wt}$ transits the preference \(\bm{w}_i[t]\) to \(\bm{w}_i[{t+1}]\). Its setting is divided in two cases in the following sections.
\(P_0\) is the distribution function of the initial state of the environment. Finally, each agent \(i\) obtains vectorised rewards as a function of the state and joint action \(\boldsymbol{r}_i \in \bm{R}: S \times \mathcal{A}_1 \times \ldots \times \mathcal{A}_N \rightarrow \bm{R}:\mathbb{R}^m\). The objective of MAMO optimisation is the optimal joint policy which maximises the weighted sum of all agents' expected rewards and their corresponding preferences: \(R = \sum_{i=0}^{N}\sum_{t=0}^{T} \gamma^t \bm{w}_i^\top \bm{r}^t_i\), where \(\gamma \in [0, 1]\) denotes the discount factor.  



\subsection{The Necessity of Global Preferences in POMOMDP Decision‐Making}

In this section, we theoretically demonstrate that the attainability of BNE in POMOMDP depends on the observability or structural modeling of global preference.

\subsubsection{Case I: Preferences as Unstructured Random Variables}
Assume $\bm{w}_i \sim \text{Unif}(\Delta^k)|_{\Delta^k = {\boldsymbol{w} \in \mathbb{R}^k \mid \sum_j w_j = 1, w_j \ge 0}}$, where each agent's preference is independently and uniformly distributed.

\begin{theorem}[BNE Inapplicability with Unobservable, Uniform Preferences]
\label{thm:unstructured-unobservable}
Suppose that for any $i\neq j$, agent $i$ knows only that other agent's preference $\bm w_j\sim\mathrm{Unif}(\Delta^k)$. Then the classical BNE concept is inapplicable
\end{theorem}

\begin{theorem}[BNE Attainability with Observable Uniform Preferences]
\label{thm:structured-observable}
Let each agent \(j\)’s preference weight $\bm w_j \;\sim\;\mathrm{Unif}(\Delta^k)$ be drawn independently, and assume that for every pair \(i\neq j\), agent \(i\) \emph{observes} \(\bm w_j\) prior to choosing its action.  Then BNE in behavioral strategies exists.
\end{theorem}

\subsubsection{Case II: Preferences as State-dependent Functions}
\begin{theorem}[BNE Existence when \(\bm w_i=g(o_i)\)]
\label{thm:obs-dependent-preferences}
Suppose each agent \(i\) ’s preference weight \(\bm w_i\) is a \emph{deterministic} function of its private observation $\bm w_i = g(o_i)$ where \(g:O_i\to\Delta^k\) is continuous.  Then, under the usual compactness and continuity assumptions on observations and actions, a mixed‐strategy BNE exists.
\end{theorem}

\begin{proof}
The proofs of the above theorems are provided in detail in Appendix A.
\end{proof}




\section{General Multi-Agent Multi-objective Reinforcement Learning}
In a game with \(N\) agents, each agent has $M$ objectives, and the corresponding preference vector \(\bm{W} = \{\bm{w}_1, \ldots, \bm{w}_N\} \) indicates the importance of each objective for the corresponding agent, where \(\bm{w}_i = \{w_i^1, \ldots, w_i^M |\sum_{j=1}^{M} w_i^j = 1\}\). Based on the above theorem, we further develop distinct MAMORL frameworks designed for whether preferences are modeled as unstructured random variables or observation-dependent functions. These frameworks are designed to generalise across diverse real-world settings of states, actions, preferences, and rewards, thereby enabling robust optimisation in MAMOS scenarios.

\subsection{Global-preference-based MAMORL for Case I}
The policy set \(\bm{\pi} = \{\pi_1, \ldots, \pi_N\}\) is assigned to all agents and is parametrised by \(\bm{\theta^\pi} = \{\theta^{\pi_1}, \ldots, \theta^{\pi_N}\}\). According to Theorem \ref{thm:structured-observable}, the global preference $\bm{W}$ is necessary for all agents to reach BNE. Thus, the input is the observation $o_i$ achieved from the state $s$ and the global preference $\bm{W}$. The probability of each action $\pi_i(a_i | o_i, \bm{W})$ is the output. 

\(\bm{v}_i^{\pi_i}: \mathbb{R}^m\) is the vectorised MO state-value of the policy \(\pi_i\), which approximates the expected rewards under the initial state distribution \(P_0\), given the actions $A_0$ and the given preference \(\bm{w}\), denoted as:
\begin{equation}
\begin{split}
  \bm{v}_i^{\pi_i} &= \mathbb{E}_{s_0 \sim P_0} \left[ \bm{q}_i^{\pi_i} (s_0, \bm{A}_0, \bm{W}) \right]\\\bm{A}_0 &= \{\pi_1(o_1[0], \bm{W}),..,\pi_N(o_N[0], \bm{W})\},{o_i[0]=P_o^i(s_0)}.
  \end{split}
\end{equation}

This MO state-value vector can be linearly combined with the preference \(\bm{w}_i\): \(v_i^{\pi_i}= \bm{w}_i^\top \bm{v}_i^{\pi_i}
\). The objective of each agent is to find the policy $\pi_i$ which maximises \(v_i^{\pi_i}\) under any given preference \(\bm{w}\). 

The vectorised MO action-value function for the policy \(\pi_i\) based on the state-action-preference tuple $((s, \bm{A}, \bm{W}))|\bm{A} = \{a_i,..,a_N\}$ is utilised to approximate the expected rewards under the policy \(\pi_i\), which is defined as Eq. \ref{MOQ}.
\begin{equation}
\begin{split}
    \bm{q}_i^{\pi_i}(s, \bm{A}, \bm{W}) &= \mathbb{E}_{\pi_i} [ \sum_{t=0}^{\infty} \gamma^t \bm{r}_i(s[t], \bm{A}[t])],\\  s[0]=s, \bm{A}[0]&=\{a_1,..,a_N\},a_i[t+1] = \pi_i(o_i[t], \bm{W}[t]) \\\bm{W}[0]&=\{\bm{w}_1,..,\bm{w}_N\|\bm w_i \;\sim\;\mathrm{Unif}(\Delta^k)\},
    \label{MOQ}
\end{split}
\end{equation}
where \(\bm{q}_i^{\pi_i}(s, \bm{A}, \bm{W})\) is an \(m\)-dimensional vector representing expected rewards of \(m\) objectives for agent \(i\). It extends the MO state-value vector by explicitly incorporating the current action; it can be directly optimised in policy learning:
$\bm{q}_i^{\pi_i}(s,\bm{A},\bm{W})
= \mathbb{E}_{s^{\prime} \sim P(\cdot \mid s, \boldsymbol{A})}\left[\boldsymbol{r}_i(s, \boldsymbol{A})+\gamma \boldsymbol{v}_i^{\pi_i}\left(s^{\prime}, \boldsymbol{W}\right)\right]$.

A centralised trained MO action-value function \(\bm{Q}_i^{\pi_i}(s, a_1, \ldots, a_N, \bm{W}|\theta^{\bm{Q}_i})\) parameterised by \(\theta^{\bm{Q}_i}\) is deployed to represent $\bm{q}_i^{\pi_i}(s, \bm{A}, \bm{W})$. The inputs are the actions of all agents \(a_1, \ldots, a_N\), the preference settings of all agents \(\bm{W}\), and the state information \(s\). It outputs represent the approximate expected rewards $\bm{v}_i^{\pi_i}$. 

The policy of each agent \(\pi_i\) is updated by the gradient of the expected return $J\left(\theta^{\pi_i}\right)=\mathbb{E}_{s_t, a_t \sim \pi}\left[\boldsymbol{w}_i^{\top} \boldsymbol{r}_i\left(s_t, a_t\right)\right]$ aimed at maximising the weighted sum of the approximate expected rewards $\bm{Q}_i^{\pi_i}$ and the current preference $\bm{w}_i$, which is represented as:
\begin{equation}
\begin{aligned}
 \nabla_{\theta^{\pi_i}} J(\theta^{\pi_i}; \bm{W}) = & \mathbb{E}_{s \sim \rho^\pi, a_i \sim \pi_i} [ \nabla_{\theta^{\pi_i}} \log \pi_i(a_i \mid o_i, {W}) \\ & \bm{w}_i^\top \bm{Q}_i^{\pi_i}(s, a_1, \ldots, a_i, \ldots, a_N, \bm{W} \mid \theta^{\bm{Q}_i}) ],
\label{MAMOPG}
\end{aligned}
\end{equation} 
where $\rho^\pi$ is the state distribution induced by the policy $\pi$. This framework can also be extended to deterministic policies. The policy is reformulated as the continuous action version: \(\bm{\mu} = \{\mu_1(o_1, \bm{W} \mid \theta^{\mu_1}), \ldots, \mu_N(o_N, \bm{W}\mid \theta^{\mu_N})\}\). The corresponding MAMO Deep Deterministic Policy Gradient (MAMODDPG) is denoted as:
\begin{equation}
\begin{split}
\nabla_{\theta^{\mu_i}} J(\theta^{\mu_i}; \bm{W}) = \mathbb{E}_{s, a, \bm{W} \sim \mathcal{D}} [ \nabla_{a_i} \bm{w}_i^\top \bm{Q}_i^{\mu_i} (s, a_1, \ldots, a_i, \ldots, \\ a_N, \bm{W} \mid \theta^{\bm{Q}_i}) \mid_{a_i = \mu_i(o_i, \bm{W} \mid \theta^{\mu_i})} \nabla_{\theta^{\mu_i}} \mu_i (o_i, \bm{W} \mid \theta^{\mu_i})],
\end{split}
\label{MAMODDPG}
\end{equation}
where the experience replay buffer \(\mathcal{D}\) contains tuples \((s, s', a_1, \ldots, a_N, \bm{W}, \bm{r}_1, \ldots, \bm{r}_N)\). By sampling the experiences from $\mathcal{D}$, all agents' deterministic actor networks are updated by maximising the MAMODDPG, and all agents' critic networks are updated by minimising the MO temporal difference (MOTD) error for the more accurate approximation:
\begin{equation}
\begin{split}
L(\theta^{\bm{Q}_i}) = &\mathbb{E}_{s, s', \bm{a}, \bm{W}, \bm{r} \sim \mathcal{D}} \\ &\left[ \bm{Q}_i^{\mu_i} (s, a_1, \ldots, a_i, \ldots, a_N, \bm{W} \mid \theta^{\bm{Q}_i}) - \bm{y}_i \right]^2.
\end{split}
\label{MOTDE}
\end{equation}

\begin{equation}
\begin{split}
\bm{y}_i = & \bm{r}_i + \gamma \bm{Q}_i^{'\mu_i'} (s', a_1', \ldots, a_i^{GPI}, \ldots, a_N', \bm{W} \mid \theta^{\bm{Q}_i'}) \\ & \mid_{a_j' = \mu_j'(o_j', \bm{W}), a_i^{GPI} = \mu_i^{GPI}(o_i', \bm{W})},
\end{split}
\label{y}
\end{equation}
where \(\bm{\mu}' = \{\mu_1', \ldots, \mu_N'\}\) and \(\bm{Q}_i^{'\mu_i'}\) are the target actor networks and critic networks for all agents maintained during the centralised training. \(\mu_i^{GPI}\) is the agent's policy integrated with Generalised Policy Improvement \cite{yang2019generalized}, which assists in the rapid exploration of the entire preference space.
In GPI, an alternative policy set for each agent \(\Pi_i\) is maintained.
All policies in \(\Pi_i\) are used to generate multiple actions, and the one with the maximum expected return \(Q_{max}^{\pi_i^*}(s, a)\) is selected by the agent $i$. For the deterministic policy and MAMO context, it is reformulated as:
\begin{equation}
\begin{split}
\mu_i^{GPI}(o_i, \bm{W}) = & \mu_i (o_i, \arg \max_{\boldsymbol{w}'_i \sim\bm{\Psi}_i} \boldsymbol{w}_i^\top \bm{Q}^{\mu_i}_i(o_i, a_1', \ldots, \\ &\mu_i(o_i, \bm{W}'), \ldots, a_N', \bm{w}_1, \ldots, \bm{w}'_i, \ldots, \bm{w}_N)).
\end{split}
\end{equation}

The policy set \(\Pi_i\) is replaced by policies generated with the random global preferences \(\bm{W}'=\{\bm{w}_1, \ldots, \bm{w}'_i, \ldots, \bm{w}_N\}\) , where other agents' preferences are fixed while the preference of itself is randomly sampled from the preference distribution $\bm{\Psi}_i$. The critic network \(\bm{Q}^{\mu_i}\) generates the expected action-value vectors from different preferences. After the weighted sum with the given preference \(\bm{w}_i\), the maximum summation is selected, and the associated action \(\mu_i(o_i, \boldsymbol{W}^*)\) is selected as the optimal action.

For the optimisation for unstructured random preference, the global preference ensures each agent maintains a consistent belief over the global joint policy during decision-making. This mechanism facilitates consensus among agents toward maximising the global utility, mitigating the non-stationarity in the evolving joint MO policy, and enabling the BNE. The following experiments demonstrate the global preference leads to an improvement in overall utility. Such a design is appropriate in scenarios where global coordination is critical and communication is available, such as swarm robotics or UAV formations.

\subsection{Agent-Attention MAMORL for Case II}
The global preference above inevitably violates the principle of decentralised execution in the CTDE paradigm, introducing non-negligible communication overhead into distributed systems. In scenarios where inter-agent communication is entirely infeasible, such as disaster-response missions in disastrous environments, this approach becomes incompatible with the constraints of fully decentralised systems. Consequently, we observe that assigning preferences as completely random variables departs from several real-world applications, where agent preferences are often shaped by environmental states or agents' observations.
\begin{figure*}[!t] 
	\centering
	\includegraphics[width=6.2in]{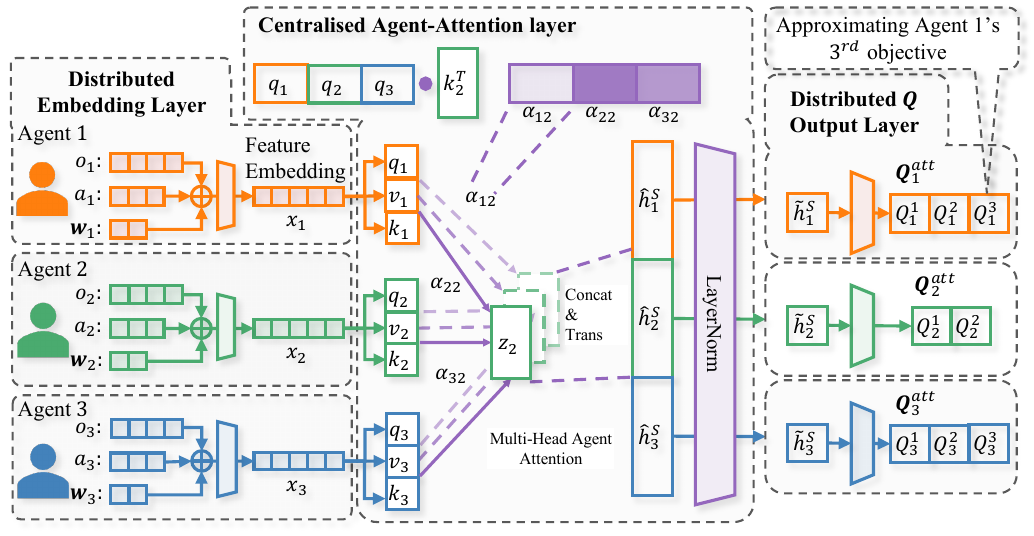}
	\caption{Agent Attention-based MAMORL Framework. Each agent uses its individual embedding layer to extract feature embeddings $x_i$ based on $o_i$, $a_i$, and $\bm{w}_i$. All agents’ embeddings are concatenated and fed into a centralised agent-attention layer. In this layer, each agent’s preference and policy are accessible to others, enabling agents to model the influence of others’ preferences and associated policies on their own policies and rewards. The output of this layer, $\hat{H}^S$, is sliced to obtain the feature corresponding to agent $i$: $\tilde{h}^S_i$ through the LayerNorm layer, which is then passed to agent $i$’s Q output layer to produce a vectorised Q-value $\mathbf{Q}^{att}_i$ that approximates its multiple rewards.}
	\label{fig2}
\end{figure*}
To address this gap, we then design an agent-attention MAMORL (AA-MAMORL) framework that better aligns with practical scenarios where preferences are observation-dependent in this section. This framework maintains the integrity of CTDE, while allowing all agents to converge toward a BNE and jointly optimise the global utility in a scalable and communication-efficient manner.

Based on theorem 3, when preference $\bm{w}_i$ is a deterministic function of the agent's observation $\bm{w}_i=g_i(o_i)$, the global decision is no longer needed for the decision. Thus the policy for each agent is defined as $\pi_i(a_i | o_i)$. And it is updated by the gradient of the expected return \(J(\theta^{\pi_i}) = \mathbb{E}[R_i]\), which is represented as
\begin{equation}
\begin{aligned}
 \nabla_{\theta^{\pi_i}} J(\theta^{\pi_i}) = & \mathbb{E}_{s \sim \rho^\pi, a_i \sim \pi_i,\bm{w}_i =g_i (o_i)} [ \nabla_{\theta^{\pi_i}} \log \pi_i(a_i \mid o_i) \\ & \bm{w}_i^\top \bm{Q}_i^{\pi_i}(s, a_1, \ldots, a_i, \ldots, a_N, \bm{W} \mid \theta^{\bm{Q}_i})[i] ].
\label{MAMOPG2}
\end{aligned}
\end{equation}

This framework can also be extended to deterministic policies. The policy is reformulated as the continuous action version: \(\bm{\mu} = \{\mu_1(o_1 \mid \theta^{\mu_1}), \ldots, \mu_N(o_N\mid \theta^{\mu_N})\}\), which is updated by the MAMO Deep Deterministic Policy Gradient (MAMODDPG):
\begin{equation}
\begin{split}
\nabla_{\theta^{\mu_i}} J(\theta^{\mu_i}) = \mathbb{E}_{s, a \sim \mathcal{D},\bm{w}_i=g_i(o_i)} [ \nabla_{a_i} \bm{w}_i^\top \bm{Q}_i^{\mu_i} (s, a_1, \ldots, \\a_i, \ldots,  a_N, \bm{W} \mid \theta^{\bm{Q}_i}) \mid_{a_i = \mu_i(o_i \mid \theta^{\mu_i})} \nabla_{\theta^{\mu_i}} \mu_i (o_i \mid \theta^{\mu_i})].
\end{split}
\label{MAMODDPG2}
\end{equation}

According to Theorem 2, the structural modeling of global preference and policy is necessary. Thus,
for the critic network, let $h^S = \{h^S_1, \dots, h^S_N\}$ denote the set of feature embeddings from $N$ agents, where each feature embedding $h^S_i$ combines the observation $o_i \in \mathbb{R}^{D_o}$, agent-specific actions $a_i$ , and local preference $\bm{w_i}$: $h^S_i = [o_i; a_i; \bm{w}_i]$

Each embedding is mapped into a common embedding space of dimension $d$ through the linear encoder layer of each agent, forming $\mathbf{x}_i \in \mathbb{R}^d$ that corresponds to the $i$-th agent's latent feature.

To effectively capture inter-agent influence under a given state and the global preference which is inferrable, and associated joint policy, we employ an agent-level attention mechanism. Unlike traditional attention applied to language or vision \cite{han2024demystify} where tokens or patches are arranged with spatial or sequential prior, here each $\mathbf{x}_i$ represents an independent and potentially heterogeneous agent's state, action, and preference.
This attention module is shared across all agents. It serves as an explicit relational reasoning mechanism, enabling each agent to dynamically incorporate inter-agent influences based on learned preference-specific and policy-specific patterns, thereby modeling agent-aware relational dependencies in a fully differentiable manner.


To capture the relational dependencies between agent $i$ and all others (including itself), the model first transforms the agent's embedding via learned projection matrices into query, key, and value vectors:
\begin{equation}
\mathbf{q}_i = \mathbf{x}_i \mathbf{W}^Q, \quad \mathbf{K} = \mathbf{X} \mathbf{W}^K, \quad \mathbf{V} = \mathbf{X} \mathbf{W}^V,
\end{equation}
where $\mathbf{W}^Q, \mathbf{W}^K, \mathbf{W}^V \in \mathbb{R}^{d \times d_h}$ are shared across agents within a given attention head, $\bm{X}$ is the collective embedding of all $N$ agents. Each agent thus uses its own query vector $\mathbf{q}_i$ to compute attention weights over all $N$ key vectors:
\begin{equation}
\boldsymbol{\alpha}_i = \text{softmax}\left( \frac{\mathbf{q}_i \mathbf{K}^\top}{\sqrt{d_h}} \right) \in \mathbb{R}^{1 \times N},
\end{equation}
which reflects how much agent $i$'s utility and policy attend to other agents when updating its representation.

The resulting embedding for agent $i$ under a single attention head is computed as:
\begin{equation}
\mathbf{z}_i = \boldsymbol{\alpha}_i \mathbf{V} = \sum_{j=1}^N \alpha_{ij} \mathbf{v}_j.
\end{equation}

The interaction weights $\alpha_{ij}$ adaptively quantify how much agent $j$'s utility function and associated policy affect agent $i$'s rewards, which is critical in maintaining the joint policy stationary in MAMOSs.

To enrich the model's expressiveness, multiple such attention heads are used in parallel. Each head independently projects the inputs using distinct parameters, producing head-specific embeddings $\mathbf{z}_i^{(1)}, \dots, \mathbf{z}_i^{(h)}$. These are then concatenated and linearly transformed to obtain the final output embedding:
\begin{equation}
\hat{h}^S[i] = {Concat}\left( \mathbf{z}_i^{(1)}, \dots, \mathbf{z}_i^{(h)} \right) \mathbf{W}^O, \quad \mathbf{W}^O \in \mathbb{R}^{hd_h \times d}.
\end{equation}

All agent's final embedding of different heads will be concatenated: $\hat{H}^S=Concat\left( \hat{h}^S_1, \dots, \hat{h}^S_N \right)$. Subsequently, a feed-forward network (FFN) with residual connections and layer normalisation further enhances representational capacity:

\begin{equation}
\tilde{h}^S = \text{LayerNorm}(\hat{h}^S + \text{FFN}(\hat{h}^S)).
\end{equation}

Finally, the slice of each agent $\tilde{h}^S_i$ is processed by the individual output layer to estimate agent-specific Q-values, enabling accurate approximation of vectorised action values under complex, multi-agent dynamics for agent $i$:

\begin{equation}
\bm{Q}_i^{att} = Output(\tilde{h}^S_i) .
\end{equation}

And the attention-based MOTD error becomes:
\begin{equation}
L^{att}(\theta^{\bm{Q}_i}) = \mathbb{E}_{s, s', \bm{a}, \bm{W}, \bm{r} \sim \mathcal{D}} \left[\bm{Q}_i^{att}(s,A,\bm{W})- \bm{y}^{att}_i \right]^2.
\label{MOTDE}
\end{equation}
\begin{equation}
\bm{y}^{att}_i =  \bm{r}_i + \gamma \bm{Q}_i^{'att}(s',A',W')  \mid_{a_j' = \mu_j'(o_j'), a_i^{GPI} = \mu_i^{GPI}(o_i')}.
\label{y}
\end{equation}

 By minimising the attention-based MOTD error, the parameters of each agent’s embedding and output layers, as well as the shared agent-attention module, are jointly updated. This enables each agent to better learn how variations in its own policy affect the vectorised rewards under the global utility functions and associated joint policies. Such relational modeling facilitates utility-aware policy improvement, guiding agents toward both global utility maximisation and convergence to a BNE. The pseudocode of these two learning frameworks can be found in Appendix B.

\begin{table*}[t]
\begingroup
\fontsize{9pt}{11pt}\selectfont
\renewcommand{\arraystretch}{1.2}
\centering
\caption{Performance comparison in 9 MAMO environments. Results are reported as mean ± standard deviation over 10 seeds.}
\label{tab:ip-hv}
\begin{tabular}{l|ccccccc}
\toprule
\textbf{Env} & \textbf{Metric} & \textbf{GPIPD} & \textbf{MOMIX} & \textbf{IP} & \textbf{MADDPG} & \textbf{AA} & \textbf{GP} \\
\midrule
\multirow{2}{*}{Catch} & GU & $315.2 \pm 3.9$ & $82.0 \pm 2.1$ & $161.0 \pm 2.4$ & $397.8 \pm 2.7$ & $\mathbf{528.1 \pm 26.7}$ & $255.9 \pm 91.4$ \\
                         & HV & $251.6 \pm 43.0$ & $2.8 \pm 4.0$ & $93.5 \pm 3.9$ & $124.1 \pm 15.7$ & $\mathbf{215.8 \pm 22.9}$ & $78.4 \pm 22.7$ \\
\midrule
\multirow{2}{*}{Escort} & GU & $269.0 \pm 66.3$ & $67.6 \pm 2.9$  & $290.0 \pm 3.4$ & $316.0 \pm 7.4$ & $\mathbf{613.8 \pm 90.8}$ & $569.9 \pm 0.6$ \\
                         & HV & $385.1 \pm 43.8$ & $12.9 \pm 3.4$& $137.6 \pm 11.7$ & $46.4 \pm 7.7$ & $\mathbf{184.8 \pm 75.0}$ & $133.8 \pm 37.2$ \\
\midrule
\multirow{2}{*}{Walker} & GU & $-103.3 \pm 0.2$ & $-99.3 \pm 0.4$ & $-102.4 \pm 0.0$ & $-100.9 \pm 0.1$ & $\mathbf{-25.8 \pm 6.8}$ & $-31.4 \pm 21.1$ \\
                         & HV & $12.6 \pm 6.6$ & $22.0 \pm 4.3$ & $15.9 \pm 3.0$ & $19.9 \pm 9.0$ & $\mathbf{3345.7 \pm 475.6}$ & $2115.7 \pm 511.6$ \\
\midrule
\multirow{2}{*}{Sur} & GU & $405.2 \pm 1.1$ & $254.6 \pm 19.6$ & $225.7 \pm 33.8$ & $405.4 \pm 0.3$ & $\mathbf{615.7 \pm 58.9}$ & $436.1 \pm 1.2$ \\
                         & HV & $70.9 \pm 13.9$ & $143.6 \pm 24.7$ & $88.1 \pm 29.6$ & $261.2 \pm 13.9$ & $\mathbf{318.4 \pm 72.0}$ & $96.3 \pm 85.5$ \\
\midrule
\multirow{2}{*}{Adv} & GU & $-155.1 \pm 11.7$ & $-61.2 \pm 11.7$ & $-160.8 \pm 34.3$ & ${-150.0 \pm 16.4}$ & $\mathbf{-26.9 \pm 2.4}$ & $-64.5 \pm 6.2$ \\
                         & HV & $276.9 \pm 29.2$ & $622.0 \pm 80.7$ & $379.7 \pm 88.4$ & $201.0 \pm 36.1$ & $\mathbf{963.0 \pm 43.4}$ & $707.5 \pm 77.6$ \\
\midrule
\multirow{2}{*}{Push} & GU & $-83.4 \pm 3.5$ & $-44.3 \pm 33.5$ & $-216.6 \pm 54.2$ & $-202.1 \pm 19.2$ & $\mathbf{-21.1 \pm 1.9}$ & $-40.6 \pm 15.7$ \\
                         & HV & $726.5 \pm 101.7$ & $658.3 \pm 69.7$ & $848.0 \pm 57.6$ & $741.1 \pm 79.0$ & $\mathbf{2183.1 \pm 469.9}$ & $1382.6 \pm 87.3$ \\
\midrule
\multirow{2}{*}{Ref} & GU & $-99.5 \pm 22.8$ & $-66.9 \pm 22.8$ & $-109.1 \pm 15.1$ & ${-265.7 \pm 431.6}$ & $\mathbf{-51.1 \pm 2.2}$ & $-67.7 \pm 9.7$ \\
                         & HV & $564.2 \pm 135.5$ & $763.0 \pm 43.9$ & $542.6 \pm 58.4$ & $488.4 \pm 78.2$ & ${749.8 \pm 91.1}$ & $\mathbf{784.2 \pm 114.8}$ \\
\midrule
\multirow{2}{*}{Spread} & GU & $-73.8 \pm 1.2$ & $-168.6 \pm 41.2$ & $-130.7 \pm 25.9$ & ${-198.0 \pm 42.8}$ & $\mathbf{-53.8 \pm 0.6}$ & $-128.3 \pm 38.7$ \\
                         & HV & $143.2 \pm 195.5$ & $278.3 \pm 25.1$ & $175.0 \pm 12.4$ & $673.1 \pm 84.6$ & $\mathbf{846.0 \pm 27.6}$ & $384.2 \pm 59.3$ \\
\midrule
\multirow{2}{*}{Tag} & GU & $-116.4 \pm 4.0$ & $-31.7 \pm 7.0$ & $-145.2 \pm 33.7$ & ${-105.1 \pm 17.2}$ & $\mathbf{-15.0 \pm 1.5}$ & $-57.8 \pm 15.4$ \\
                         & HV & $176.7 \pm 67.5$ & $358.9 \pm 24.2$ & $289.4 \pm 28.6$ & $243.9 \pm537.4$ & $\mathbf{521.8 \pm 56.2}$ & $295.1 \pm 28.6$ \\
\bottomrule
\end{tabular}
\endgroup
\end{table*}

\section{Experiment Results and Analysis}
\subsection{Experiment Settings}
\subsubsection{Datasets}
\begin{itemize}
\item \textbf{MOMA particle environments}: 
A series of environments extended from the grounded particle environment \cite{lowe2017multi} into an MO version. The first objective aligns with the original one, the second one is the energy consumption related to movement and communication. The benchmark includes the following scenarios: \textit{Push}, \textit{Adversary}, \textit{Reference}, \textit{Spread}, and \textit{Tag}.
\item \textbf{MOMALand} \cite{felten2024momaland}:  
A benchmark that builds on the PettingZoo API and supports MAMO learning by returning vector-valued rewards. It includes diverse scenarios such as \textit{Mountain Walker}, \textit{Escort}, \textit{Catch}, and \textit{Surround}. Detailed descriptions of these environments can be found in Appendix C.
\end{itemize}
 \subsubsection{Baselines}
\begin{itemize}
\item \textbf{MO-MIX} \cite{hu2023mo}: Utilises preference-conditioned local action-value estimation and a parallel mixing network to compute joint value functions. A preference-based exploration mechanism is introduced to encourage well-distributed Pareto-optimal solutions.

\item \textbf{GPI-PD} \cite{alegre2023sample}: Combines GPI with a Dyna-style MORL approach to prioritise updates for improved sample efficiency. Modifications are made to support the multi-agent setting.

\item \textbf{Individual Preference (IP)}: Each agent learns its policy based solely on its local observation and private preference vector. This can be viewed as a part of ablation tests.

\item \textbf{MADDPG} \cite{zhang2024multi}
: A standard single-objective MADDPG baseline with scalarised rewards computed from multiple objectives using current preferences. This can be viewed as a part of ablation tests.
\end{itemize}
\subsubsection{Evaluation Metrics}
\begin{itemize}
 \item \textbf{Global Utility (GU)} \cite{alegre2023sample}: Multiple objectives are weighted summed by the preference $\bm{w}_i$ to achieve the individual utility $v_i^{\pi_i}(\bm{w}_i)$ for each agent. All agent's individual utility will be averages to get the GU. We average over 128 initial states for diverse preference settings to approximate the preference space.
    \item \textbf{Hypervolume (HV)} \cite{zitzler2007hypervolume}: The volume of the area in the objective space enclosed by the reference points and the non-dominated solutions obtained by the algorithm.

 \end{itemize}

\begin{figure*}[!t]
    \centering
        \subfigure[Mountain Walker]{
        \includegraphics[width=0.49\columnwidth]{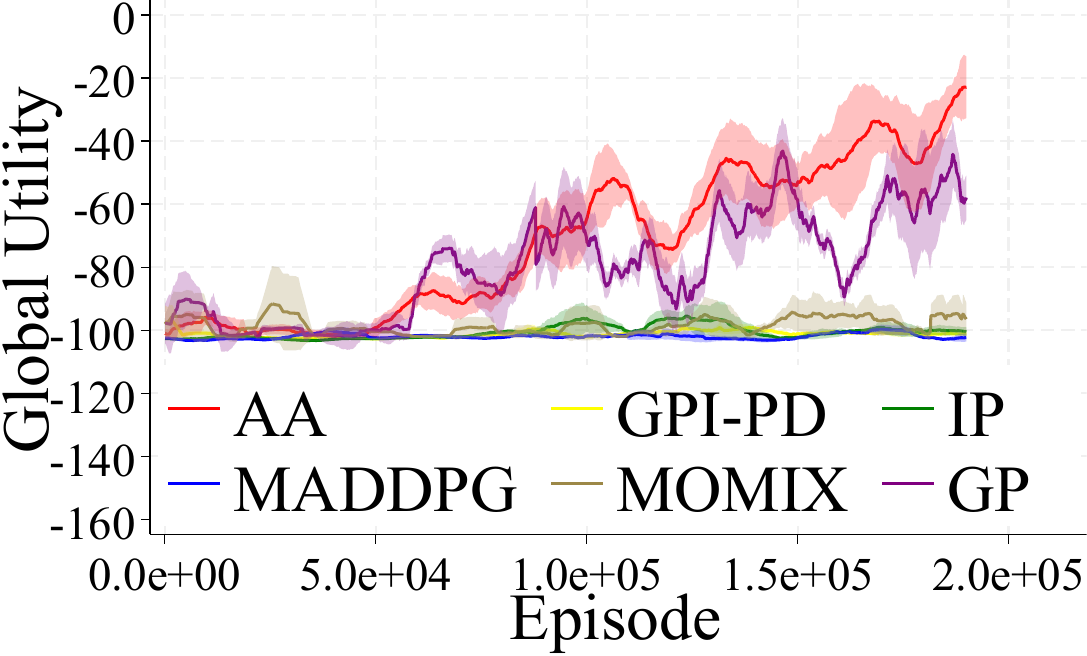}
        \label{fig:image3}
    }
     \subfigure[Surround]{
        \includegraphics[width=0.49\columnwidth]{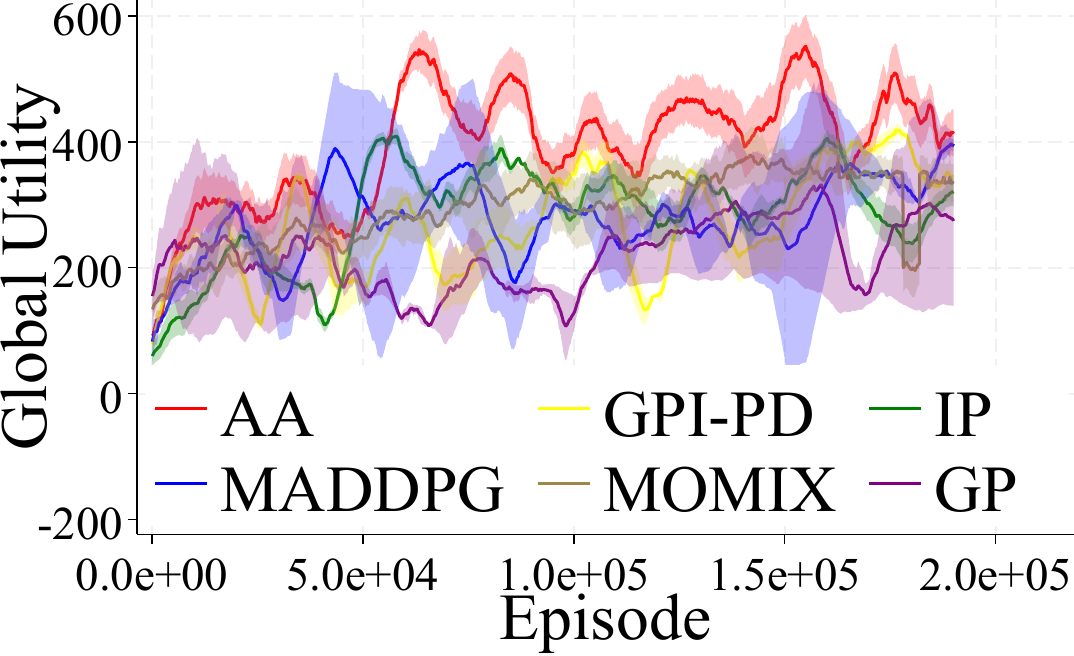}
        \label{fig:image3}
    }
    \hspace{0\textwidth} 
    \subfigure[Tag]{
        \includegraphics[width=0.49\columnwidth]{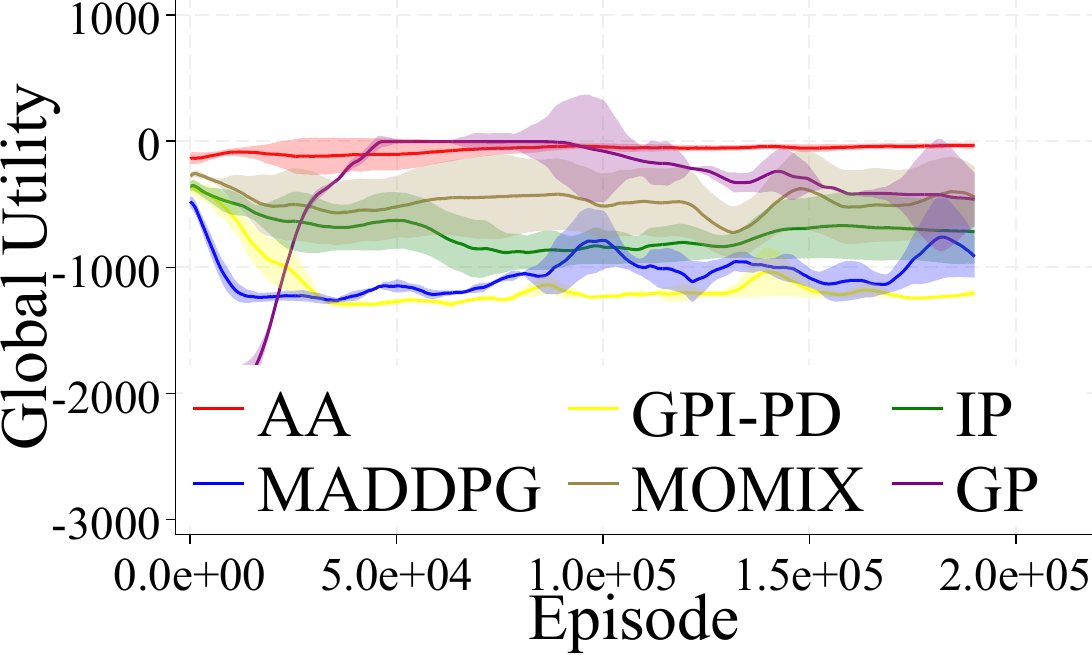}
        \label{fig:image2}
    }
    \hspace{0\textwidth} 
    \subfigure[Push]{
        \includegraphics[width=0.49\columnwidth]{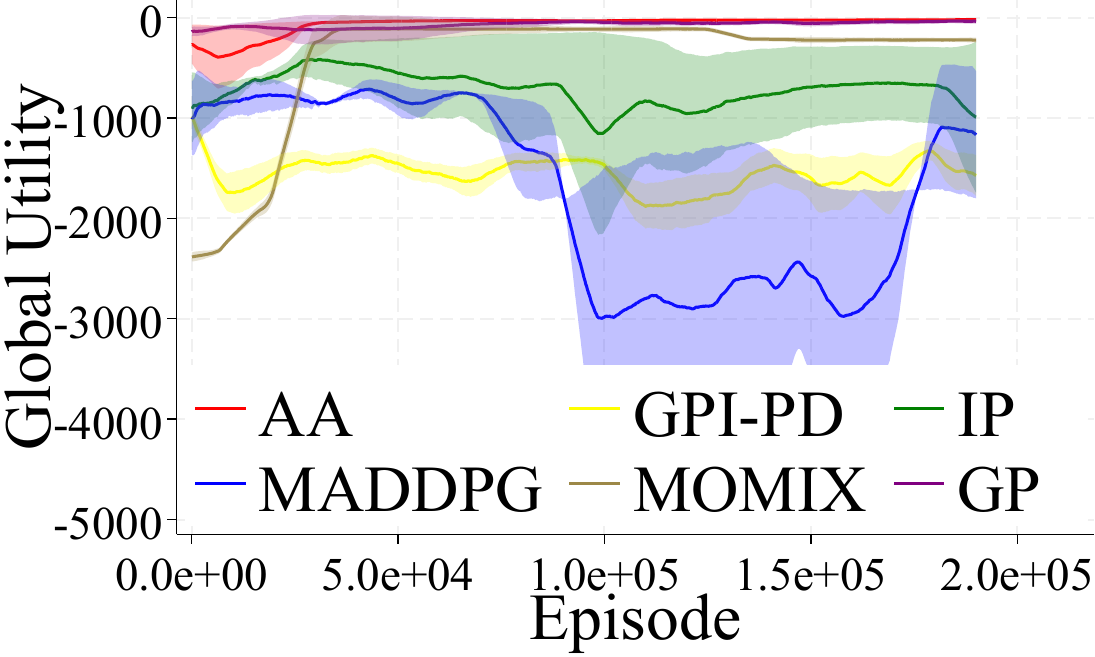}
        \label{fig:image1}
    }
    \caption{Average and 95\% confidence intervals of GU on training results from AA, GP, and baselines on 4 MAMO environments.}
    \label{SFC}
\end{figure*}


\subsection{Performance Comparison}

In this experiment, each agent’s preference is modeled as a linear function of the observation. The mapping function is agent-specific but kept consistent across all rounds and baselines to ensure fairness. Detailed hyperparameter settings are presented in Appendix D.

Table~\ref{tab:ip-hv} presents the performance across 10 training seeds, reported as the mean and standard deviation of GU and HV. AA consistently achieves the best and most stable performance across most environments. Although global preference performs comparably to AA in certain environments (\textit{Walker}, \textit{Push}, and \textit{Reference}), its lack of consensus among agents undermines its robustness in settings with conflicting individual preferences (\textit{Catch}). 
GPIPD is designed for single-agent multi-objective settings. As a result, it shows poor performance and occasionally fails to converge. This demonstrates the limitation of single-agent frameworks in modeling the multi-agent policy-preference space. MOMIX is designed for the team preference setting and discrete action space. It lacks the generalisability for individual reward and preference settings. Consequently, MO-MIX shows its shortcomings in generalisation to complex reward and preference settings.

Learning curves in Fig.~\ref{SFC} show that in the most challenging environment, \textit{Multi-Walker}, only AA and GP successfully acquire meaningful policies, whereas all other baselines fail to progress. Similar trends are observed in the MOMA particle environments, where AA and global preference demonstrate the most stable learning dynamics, while others struggle with convergence and exhibit high variance.

\subsection{Ablation Study}

The comparison among AA, global preference (GP), MADDPG, and IP serves as the ablation study, aimed at evaluating the impact of vectorised action value, global preference, and AA mechanism on MAMO learning.

The struggles of MADDPG highlight the importance of vectorised representations of rewards, action-values, and preferences in multi-objective settings. Scalarised rewards might be effective when preferences are fixed and aligned between training and execution, but they become impractical in real-world scenarios where preferences evolve dynamically. As shown in our dynamic preference setting, scalarised reward methods struggle to generalise and fail to extract meaningful policies. The poor performance of IP provides empirical support for Theorem~1. Without direct access or structured modeling to the global preference, agents cannot form accurate beliefs over others’ policies, making BNE unattainable. Consequently, learning becomes unstable and ineffective. AA even outperforms global preference in several environments. This result demonstrates that how heterogeneous and dynamic utilities influence inter-agent coordination. By explicitly modeling these relational factors, AA enables each agent to update its own policy, adapting to other policies more effectively, resulting in improved performance and convergence towards the BNE.

\section{Conclusion}
In this paper, we mathematically prove that direct access to or structured modeling of global preferences during decision-making is a necessary condition for achieving BNE in MAMOS. For the case where preferences are randomly generated, we incorporate global preferences into distributed decision-making and design a corresponding MAMORL framework. For the more realistic setting where preferences are generated by agents based on their own observations, we develop an AA-MAMORL framework in which a centralised attention-based critic network is employed to model inter-agent influences and preference-policy dependencies, and is shared among all agents. To evaluate AA and global preference, we constructed experiments using 9 standard MAMO envs. The results demonstrate that the proposed AA-MAMORL consistently outperforms baselines across diverse environments by effectively modeling heterogeneous preferences and coordinating decentralised policies. Ablation results highlight the necessity of global preference modeling and vectorised objectives for stable and convergent learning in multi-agent multi-objective settings.

\bibliography{aaai2026}
\clearpage
\appendix
\onecolumn
\section*{Appendix of Achieving Equilibrium under Utility Heterogeneity: An Agent-Attention Framework for Multi-Agent Multi-Objective Reinforcement Learning}

\subsection*{Appendix A: The Necessity of Global Preferences in POMOMDP Decision‐Making}
\subsubsection*{Bayesian Game Formulation:}
We consider a Bayesian game 
\begin{equation}
\mathcal{G}_B 
= 
\langle 
\mathcal{N}, 
\{\mathcal{X}_i\}_{i\in\mathcal{N}},
\{\Theta_i\}_{i\in\mathcal{N}},
\{u_i\}_{i\in\mathcal{N}},
\{\mu_i\}_{i\in\mathcal{N}}
\rangle,
\end{equation}
where:

\begin{itemize}
  \item $\mathcal{N} = \{1,\ldots,N\}$ denotes the set of agents.
  \item $\mathcal{X}_i$ is the (possibly continuous) action space of agent $i$, and $\mathcal{X} = \prod_{i\in\mathcal{N}} \mathcal{X}_i$ is the joint action space.
  \item $\Theta_i$ is the type space of agent $i$, where each type $\theta_i = (o_i, \bm{w}_i)$ consists of the observation $o_i$ and the individual preference vector $\bm{w}_i$.
  \item The {utility function} of agent $i$ is defined as
\begin{equation}
  u_i(x_i, x_{-i}, \theta_i)
  = 
  \bm{w}_i^{\top}\bm{f}_i(x_i, x_{-i}),
\end{equation}
  where $\bm{f}_i(\cdot)$ is the vectorised reward feedback obtained from the joint action profile $(x_i, x_{-i}) \in \mathcal{X}$.
  \item $\mu_i(\theta_{-i} \mid \theta_i)$ represents agent $i$'s {belief} about the other agents' types $\theta_{-i}$, conditional on its own type $\theta_i$.
\end{itemize}

Each agent $i$ selects a measurable strategy 
\begin{equation}
s_i : \Theta_i \rightarrow \mathcal{X}_i,
\end{equation}
which specifies its action for every possible type.

\noindent
A {Bayesian Nash Equilibrium (BNE)} is a strategy profile 
$s^* = (s_1^*, \ldots, s_N^*)$ 
such that, for all $i \in \mathcal{N}$ and all $\theta_i \in \Theta_i$,
\begin{equation}
\begin{split}
    \mathbb{E}_{\theta_{-i}\sim \mu_i(\cdot \mid \theta_i)}
\!\left[
u_i(s_i^*(\theta_i), s_{-i}^*(\theta_{-i}), \theta_i)
\right]
\ge
\mathbb{E}_{\theta_{-i}\sim \mu_i(\cdot \mid \theta_i)}
\! \\ \left[
u_i(a_i, s_{-i}^*(\theta_{-i}), \theta_i)
\right],
\quad
\forall a_i \in \mathcal{X}_i.  
\end{split}
\end{equation}

\subsubsection*{Case I: Preferences as Unstructured Random Variables:}

Assume that each agent's preference $\bm{w}_i$ is independently and uniformly distributed over the simplex $\Delta^k$, i.e., $\bm{w}_i \sim \mathrm{Unif}(\Delta^k)$.

\begin{theorem}[BNE Inapplicability with Unobservable, Uniform Preferences]
\label{thm:inapplicability}
Suppose that for any $i\neq j$, agent $i$ knows only that $\bm w_j\sim\mathrm{Unif}(\Delta^k)$ and receives no informative signal about $\bm w_j$. Then, the classical Bayesian Nash Equilibrium (BNE) concept is \emph{inapplicable}: the best–response correspondence cannot be properly defined because the conditional expectations required for expected utility are not well-posed under agent $i$'s information structure.
\end{theorem}

\begin{proof}
Fix agent $i \in \mathcal{N}$. Its type is $\theta_i=(o_i,\bm w_i)\in\Theta_i$, known to itself, and let $\theta_{-i}=(o_{-i},\bm w_{-i})\in\Theta_{-i}$ denote the types of the other agents.

In the standard BNE framework, each agent adopts a behavioral strategy
\begin{equation}
s_i:\Theta_i \rightarrow \Delta(\mathcal{X}_i),
\end{equation}
and evaluates the expected utility of a mixed action $\alpha_i \in \Delta(\mathcal{X}_i)$ given opponents’ strategies $s_{-i}$ as:
\begin{equation}
\label{eq:EU}
\begin{split}
U_i(\alpha_i, s_{-i}\mid \theta_i)
&=\mathbb{E}_{\theta_{-i}\sim\mu_i(\cdot\mid\theta_i)}
\Bigl[
\bm w_i^\top \bm f_i(\alpha_i, s_{-i}(\theta_{-i}))
\Bigr].
\end{split}
\end{equation}

Here, $\mu_i(\cdot\mid\theta_i)$ denotes agent $i$’s posterior belief about $\theta_{-i}$, and $\bm f_i(\cdot)$ represents the vectorised reward function.

Under the assumption that (i) $\bm w_j$ is unobservable to $i$, and (ii) $i$ receives no informative signal about $\bm w_{-i}$, Bayes’ rule yields the \emph{uninformative posterior}:
\begin{equation}
\mu_i(\bm w_{-i}\mid \theta_i)
=\mathrm{Unif}(\Delta^k)^{\otimes (N-1)},
\end{equation}
independent of $\theta_i$. Although the marginal distribution of $\bm w_{-i}$ is known, the mapping $s_{-i}:\Theta_{-i}\to\Delta(\mathcal{X}_{-i})$ is not measurable with respect to the $\sigma$-algebra generated by $\theta_i$. Consequently, $i$ has no well-defined information basis to form beliefs about the random variable $s_{-i}(\theta_{-i})$.

As a result, the integrand in \eqref{eq:EU},
\begin{equation}
\bm f_i(\alpha_i, s_{-i}(\theta_{-i})),
\end{equation}
is not measurable with respect to $\sigma(\theta_i)\otimes\mathcal{B}(\Theta_{-i})$, and hence the Bochner integral defining $U_i(\alpha_i, s_{-i}\mid \theta_i)$ does not exist. Without a well-defined expected utility, the best–response correspondence
\begin{equation}
BR_i(s_{-i}\mid \theta_i)
=\arg\max_{\alpha_i\in\Delta(\mathcal{X}_i)}
U_i(\alpha_i, s_{-i}\mid \theta_i)
\end{equation}
cannot be constructed. Since the classical BNE definition requires each $BR_i$ to be well-defined (measurable, convex-valued, and upper hemicontinuous), the notion of BNE itself becomes inapplicable under this unobservable, unstructured preference scenario.

\end{proof}

\subsubsection*{Case II: Preferences as Observable Random Variables:}

\begin{theorem}[BNE Attainability with Observable Uniform Preferences]
\label{thm:structured-observable}
Let each agent’s preference $\bm w_j \sim \mathrm{Unif}(\Delta^k)$ be drawn independently, and suppose that for all $i\neq j$, agent $i$ can observe $\bm w_j$ before choosing its action. Then, under standard compactness and continuity assumptions, a Bayesian Nash equilibrium in behavioral strategies exists.
\end{theorem}

\begin{proof}
When all preferences $\bm W=\{\bm w_j\}_{j=1}^N$ are common knowledge, each agent $i$’s only private information is its local observation $o_i$. The belief function thus reduces to:
\begin{equation}
\mu_i(o_{-i}\mid o_i, \bm W)=\mu_O(o_{-i}\mid o_i).
\end{equation}
The game then becomes a standard incomplete-information game where each agent’s type is $o_i$, and its utility is
\begin{equation}
u_i(x_i, x_{-i}; o_i, \bm W)
=\bm w_i^\top \bm f_i(x_i, x_{-i}).
\end{equation}
Under the assumptions that each $\mathcal{X}_i$ is compact and convex, and $\bm f_i$ is continuous and bounded, the expected payoff
\begin{equation}
\begin{split}
    U_i(\alpha_i, s_{-i}\mid o_i, \bm W)
=\mathbb{E}_{o_{-i}\sim\mu_O(\cdot\mid o_i)}\Bigl[
\mathbb{E}_{x_i\sim\alpha_i, x_{-i}\sim s_{-i}(o_{-i})} \\
[\bm w_i^\top\bm f_i(x_i, x_{-i})]
\Bigr]
\end{split}
\end{equation}
is continuous in both $\alpha_i$ and $s_{-i}$.

Each agent’s behavioral strategy $s_i:O_i\to\Delta(\mathcal{X}_i)$ defines a compact and convex strategy space $S_i$ under the weak$^*$ topology. Standard results (Berge’s maximum theorem and Kakutani/Glicksberg fixed-point theorem) imply that the aggregate best–response correspondence
\begin{equation}
BR:S\rightrightarrows S, \qquad BR(s)=\prod_i BR_i(s_{-i})
\end{equation}
is nonempty, convex-valued, and upper hemicontinuous. Hence, there exists a fixed point $s^*\in S$ such that $s^*\in BR(s^*)$, which is a Bayes–Nash equilibrium.
\end{proof}

\subsubsection*{Case III: Preferences as State-dependent Functions:}

\begin{theorem}[BNE Existence under State-dependent Preferences]
\label{thm:obs-dependent-preferences}
Suppose each agent’s preference weight $\bm w_i$ is a deterministic continuous function of its private observation, $\bm w_i=g(o_i)$ with $g:O_i\to\Delta^k$ continuous. Then, under compactness and continuity of $\mathcal{X}_i$ and $\bm f_i$, a mixed-strategy Bayesian Nash equilibrium exists.
\end{theorem}

\begin{proof}
Under $\bm w_i=g(o_i)$, each agent’s private type simplifies to $\theta_i=o_i$. The utility of agent $i$ becomes:
\begin{equation}
u_i(x_i, x_{-i}; o_i)=g(o_i)^\top\bm f_i(x_i, x_{-i}),
\end{equation}
which is continuous in $(x_i,x_{-i},o_i)$ and bounded. The existence of BNE then follows directly from the same distributional strategy and fixed-point arguments as in Theorem~\ref{thm:structured-observable}.
\end{proof}

\subsection*{Appendix B: Pseudocode for Global preference and Agent Attention MAMORL}
\subsubsection*{Global-Preference-Based Multi-Agent Multi-Objective Reinforcement Learning:}

In Global-preference-based MAMORL, the input includes the global preference distribution set \(\bm{\Psi}_{1..N}\) for all agents, the soft update coefficient \(\lambda\), and a set of agent-specific preference generators \(PG_{1..N}[o_{1..N}]\) (Line~1). Each agent \(i\) initialises its critic network \(\bm{Q}_i^\mu(s,a_1,...,a_N,\bm{W}|\theta^{\bm{Q}_i})\) and actor network \(\mu_i(o_i,\bm{W}|\theta^{\mu_i})\), where both networks condition on the global preference \(\bm{W}\) (Line~2). The parameters of the target actor and critic networks are copied from their respective networks to support stable learning through soft updates (Line~3). A shared replay buffer \(R\) is also initialised to store experience tuples (Line~4).

For each episode, the environment state \(s[0]\) is initialised(Lines~6). During each timestep \(t\), every agent \(i\) observes its local state \(o_i[t]\) based on the global state \(s[t]\), and generates a preference vector \(\bm{w}_i[t]\) via its preference generator \(PG_i\), which is stored in the global preference set \(\bm{W}[t]\) (Lines~8--11).

Then, each agent selects an action with the policy \(\mu_i(o_i[t], \bm{W}[t])\)  incorporated with the adjustable noise $\mathcal{N}$ (Lines~12--14). The environment returns a vectorised reward \(\bm{r}[t]\) and the next state \(s[t+1]\) after all actions have been executed, and the full transition \(\{s[t], \bm{a}[t], \bm{r}[t], s[t+1], \bm{W}[t]\}\) is appended to the replay buffer \(R\) (Lines~15--16).

During training, if the update condition is met, a mini-batch of \(N_{\tau}\) transitions is uniformly sampled from the buffer \(R\) (Line~18). For each sampled transition \(\{s[k], \bm{a}[k], \bm{r}[k], s[k+1], \bm{W}[k]\}\), every agent \(i\) updates its critic parameters \(\theta^{\bm{Q}_i}\) by minimising the multi-objective temporal-difference error (MOTDE) loss (Line~20), and updates its actor parameters \(\theta^{\mu_i}\) by applying the gradient of the MAMODDPG loss (Line~22).

To ensure stable training, soft updates are performed on the target critic and actor networks using the soft coefficient \(\lambda\): (Lines~25--28). Additionally, the noise $\mathcal{N}$ is reduced as the training progresses (Lines~29). This learning process is repeated over episodes until convergence.
\subsubsection*{Agent-Attention-Based Multi-Agent Multi-Objective Reinforcement Learning:}

In the proposed Agent Attention MAMORL framework, the input includes the global preference distribution set \(\bm{\Psi}_{1..N}\), the soft update parameter \(\lambda\), and the preference generator \(PG_{1..N}[o_{1..N}]\) for each agent (Line~1). Each agent \(i\) initialises: (i) an embedding network \({emb}_i^\mu(o_i,a_i,\bm{w}_i|\theta^{{emb}_i})\), which encodes observation, action, and preference vectors into a latent representation; (ii) an output network \({out}_i^\mu(\tilde{h_i^S}|\theta^{{out}_i})\), which maps the attention-interacted hidden state to vectorised Q-values; (iii) an actor policy network \(\mu_i(o_i|\theta^{\mu_i})\); and (iv) a shared agent-attention network \(\bm{att}^\mu(x_1,..,x_N|\theta^{att})\) that models interactions among agent embeddings (Line~2). Their corresponding target networks are initialised by copying the parameters from the online networks (Line~3). A centralised replay buffer \(R\) is also initialised to store transition experiences (Line~4).
For each episode, the environment is reset to initial state \(s[0]\) (Lines~6). At each time step \(t\), each agent receives its local observation \(o_i[t]\) from the observation function \(P_o^i[s[t]]\) and generates its individual preference vector \(\bm{w}_i[t]\) using its preference generator \(PG_i\), which is stored in \(\bm{W}[t]\) (Lines~8--11).

Each agent then selects an action with the policy \(\mu_i(o_i[t])\) incorporated with the adjustable noise $\mathcal{N}$ (Lines~12--14). After execution, the agents jointly observe a vectorised reward \(\bm{r}[t]\) and the next global state \(s[t+1]\). The tuple \((s[t],\bm{a}[t],\bm{r}[t],s[t+1],\bm{W}[t])\) is stored in buffer \(R\) (Lines~15--16).

If the update condition is triggered, a batch of \(N_\tau\) transitions is sampled from the buffer (Line~18). For each sampled transition, every agent \(i\) computes its attention-based Q-value \({Q}_i^{att}\) and corresponding target value \(y_i^{att}\) using the embedding network \({emb}_i\), output network \({out}_i\), shared attention layer \(att\), and the target counterparts of each (Line~20-22). The loss is then computed based on a Multi-Objective Temporal Difference Error (MOTDE) using the attention-based Q-values (Line~23).

\small
\begin{algorithm}[H]
    \caption{Global-preference-based MAMORL}
    \label{alg:algorithm-label}
    \begin{algorithmic}[1]
\STATE\textbf{Input:}$\bm{\Psi}_{1..N}$: the preference distribution set for all agents; $\lambda$: the soft update parameter; $PG_{1..N}[o_{1..N}]$: the preference generator for each agent; 
\STATE Initialise each agent's critic network $\bm{Q}_i^\mu(s,a_1,...,a_N,\bm{W}|\theta^{\bm{Q}_i})$ and actor network $\mu_i(o_i,\bm{W}|\theta^{\mu_i})$ with parameters $\theta^{\bm{Q}_i}$ and $\theta^{\mu_i}$;
\STATE  Initialise target critic network $\bm{Q}^{\mu^\prime}_i$ and target actor network $\mu'_i$ with parameters $\theta^{\bm{Q'}_i} \leftarrow \theta^{\bm{Q}_i}$, $\theta^
{\mu'_i} \leftarrow \theta^{\mu_i}$;
\STATE   Initialise replay buffer $R$;
    \FOR{$episode = 1,...,M$}
    \STATE  Initialise the state $s[0]$;
      \FOR{$t = 0,...,T$}   
        \FOR{agent $i=1,...,N$}   
        \STATE Achieve the obersvation $o_i[t]$ through $P_o^i[s[t]]$;
        \STATE Achieve the preference $\mathbf{w}_i[t]$ through $PG_i[o_i[t]]$ and store it in global preference $\mathbf{W}[t]$;
        \ENDFOR
        \FOR{agent $i=1,...,N$}   
     \STATE   Select and execute action: $a_i[t] =  \mu_i(o_i,\mathbf{W}) + \mathcal{N}$
    \normalsize
    \ENDFOR
     \STATE   Observe the vectorised reward $\bm{r}[t]$ and new state $s[t+1]$;
      \STATE  Store the transition $(s[t],\bm{a}[t],\bm{r}[t],s[t+1],\bm{W}[t])$ in $R$ ;
        \IF{$update$}
      \STATE    Sample $N_{\tau}$ transitions $\thicksim R$;
          \FOR{each experience $(s[k],\bm{a}[k],\bm{r}[k],s[k+1],\bm{W}[k])$ in $N_{\tau}$}
          \FOR{agent $i=1,...,N$}
          \STATE Update $\theta^{\bm{Q}_i}$ by minimising MOTDE $L(\theta^{\bm{Q}_i})$ ;
        \STATE         Update $\theta^{\mu_i}$
                    by descending its MAMODDPG;
                    \ENDFOR
        \ENDFOR
        \FOR {agent $i=1,...,N$}
         \STATE            $ \theta^{Q_i^{\prime}} \leftarrow \lambda \theta^{Q_i}+(1-\lambda) \theta^{Q_i^{\prime}}$ ;

         \STATE            $\theta^{\mu_i^{\prime}} \leftarrow \lambda \theta^{\mu_i}+(1-\lambda) \theta^{\mu_i^{\prime}}$;
         \ENDFOR
        \STATE $\mathcal{N}$ is reduced;
        \ENDIF
        \ENDFOR
        \ENDFOR
    \end{algorithmic}
\end{algorithm}

\begin{algorithm}[H]
    \caption{Agent Attention MAMORL}
    \label{alg:algorithm-label}
    \begin{algorithmic}[1]
\STATE\textbf{Input:}$\bm{\Psi}_{1..N}$: the preference distribution set for all agents; $\lambda$: the soft update parameter; $PG_{1..N}[o_{1..N}]$: the preference generator for each agent; 
\STATE Initialise each agent's embedding network ${emb}_i^\mu(o_i,a_i,\bm{w}_i|\theta^{{emb}_i})$, output network ${out}_i^\mu(\tilde{h_i^S}|\theta^{{out}_i})$, actor network $\mu_i(o_i|\theta^{\mu_i})$ with parameters $\theta^{{emb}_i}$, $\theta^{{out}_i}$ and $\theta^{\mu_i}$, and shared agent-attention network $\bm{att}^\mu(x_1,..,x_N|\theta^{att})$
\STATE  Initialise each agent's target embedding network ${emb}_i^{'\mu}(o_i,a_i,\bm{w}_i|\theta^{{emb'}_i})$, target output network ${out}_i^{'\mu}(\tilde{h_i^S}|\theta^{{out'}_i})$, target actor network $\mu'_i(o_i|\theta^{\mu'_i})$, and shared target agent-attention network $\bm{att}^{'\mu}(x_1,..,x_N|\theta^{att'})$
\STATE   Initialise replay buffer $R$;
    \FOR{$episode = 1,...,M$}
    \STATE  Initialise the state $s[0]$;
      \FOR{$t = 0,...,T$}   
        \FOR{agent $i=1,...,N$}   
        \STATE Achieve the obersvation $o_i[t]$ through $P_o^i[s[t]]$;
        \STATE Achieve the preference $\mathbf{w}_i[t]$ through $PG_i[o_i[t]]$ and store it in global preference $\mathbf{W}[t]$;
        \ENDFOR
        \FOR{agent $i=1,...,N$}   
     \STATE   Select and execute action: $a_i[t] = \mu_i(o_i)+\mathcal{N}$
    \ENDFOR
     \STATE   Observe the vectorised reward $\bm{r}[t]$ and new state $s[t+1]$;
      \STATE  Store the transition $(s[t],\bm{a}[t],\bm{r}[t],s[t+1],\bm{W}[t])$ in $R$ ;
        \IF{$update$}
      \STATE    Sample $N_{\tau}$ transitions $\thicksim R$;
          \FOR{each experience $(s[k],\bm{a}[k],\bm{r}[k],s[k+1],\bm{W}[k])$ in $N_{\tau}$}
          \FOR{agent $i=1,...,N$}
          \STATE use ${emb}_i$,${out}_i$, shared $att$ ,and corresponding target networks to achieve $Q_i^{att}$ and $y_i^{att}$
                    \ENDFOR
          \STATE Calculate attention-based MOTDE $L^{att}(\theta^{\bm{Q}_i})$
          \FOR{agent $i=1,...,N$}
          \STATE Update $\theta^{{emb}_i}$ and $\theta^{{out}_i}$ by minimising $L^{att}(\theta^{\bm{Q}_i})$;
        \STATE       Update $\theta^{\mu_i}$
                    by descending its MAMODDPG;
                    \ENDFOR
        \STATE Update $\theta^{{att}}$ by minimising $L^{att}(\theta^{{att}})$;         
        \ENDFOR
        \FOR {agent $i=1,...,N$}
         \STATE    update $i$'s $\theta^{emb'_i}$, $\theta^{out'_i}$, and $\theta^{\mu'_i}$;
         \ENDFOR
         \STATE    update $\theta^{att'}$;
        \STATE $\mathcal{N}$ is reduced
        \ENDIF
        \ENDFOR
        \ENDFOR
    \end{algorithmic}
\end{algorithm}
    \normalsize

Each agent updates the parameters of \({emb}_i\) and \({out}_i\) by minimising the attention-based loss \(L^{att}(\theta^{\bm{Q}_i})\) (Line~25), and updates its actor network \(\mu_i\) by descending the gradient of its MAMODDPG loss (Line~26). The shared attention network \(att\) is also updated by minimising the attention loss (Line~28). 

Subsequently, target networks \({emb'}_i\), \({out'}_i\), \({\mu'}_i\), and the shared target attention layer \({att'}\) are updated using soft updates or parameter copies (Lines~30--32). Finally, the noise $\mathcal{N}$ is reduced as the training progresses (Lines~34). This process continues across episodes until convergence.

\normalsize
\subsection*{Appendix C: Environment Introduction}
\subsubsection*{Introduction on MOMA particle environments:}
The experiments were conducted partly in MAMO environments, which were extended from the grounded particle environment \cite{lowe2017multi} into an MO version. This environment consists of $N$ agents and $L$ landmarks in a two-dimensional world with continuous space and discrete time. In each discrete time slot, each agent moves according to the applied force. Meanwhile, agents can communicate with each other. The energy consumption related to movement is calculated based on the applied force and the distance moved, while the energy consumption for communication is determined by the packet length and the energy consumption per bit. 
The environments are briefly described as follows:

•  \textbf{Cooperative Navigation (\textit{Spread})}: Multiple agents cooperate to cover all landmarks while avoiding collision. The first objective is to minimise their respective distances to the landmarks, while the second objective is to minimise the total energy consumption of all agents.

•  \textbf{\textit{Reference}}: All agents communicate with each other about the correct landmark to navigate towards. The first objective is to minimise their respective distances to their own landmarks, and the second objective is to minimise their total energy consumption by communication and movement.

•  \textbf{Keep-away (\textit{Push})}: The cooperative agents aim to reach the target landmark. Their first objective is to minimise the smallest distance of any agent to the correct landmark. Adversarial agents attempt to block them without knowing the correct landmark, and their first objective is to maximise the smallest distance between cooperative agents to the correct landmark. The second objective of each agent is to minimise its own energy consumption.

•  \textbf{Predator-prey (\textit{Tag})}: Slower cooperating agents chase a faster adversary in an obstacle-filled environment. Cooperative agents receive a reward for catching the adversary, while the adversary is punished for their first objectives. The second objective of each agent is to minimise its own energy consumption.

•  \textbf{Physical Deception (\textit{Adversary})}: Cooperative agents spread out across landmarks to deceive an adversarial agent unaware of the correct target landmark. The distance of the adversarial agent to the correct landmark serves as a punishment for cooperative agents and a reward for the adversary for their first objective. The second objective of each agent is to minimise its own energy consumption.

\subsubsection{Introduction on MAMOLand:}
MOMAland is an open source Python library for developing and comparing multi-objective multi-agent reinforcement learning algorithms by providing a standard API to communicate between learning algorithms and environments, as well as a standard set of environments compliant with that API. Essentially, the environments follow the standard PettingZoo APIs, but return vectorised rewards as numpy arrays instead of scalar values \cite{felten2024momaland}.

•  \textbf{Mountain Walker}: In this environment, multiple walker agents aim to carry a package to the right side of the screen without falling. This environment also supports continuous observations and actions. The multi-objective version of this environment includes an additional objective to keep the package as steady as possible while moving it. Naturally, achieving higher speed entails greater shaking of the package, resulting in conflicting objectives. The number of agents is configurable.

•  \textbf{Surround}: Each agent perceives its own 3D coordinates, those of its teammates, and the position of a shared target. The action space consists of discrete 3D motion vectors, and episodes terminate upon collisions, floor contact, or target capture. Agents aim to establish a stable formation around a fixed target point. The first objective is minimising the distance to the target using potential-based shaping. The second one is maximising the separation from teammates to avoid collisions.

•  \textbf{Catch}: In Catch, the target exhibits adversarial intelligence. It moves away from the centroid of the swarm if agents get too close, or randomly otherwise. The same two objectives apply, with the added complexity of an evasive target behavior that requires predictive coordination.

•  \textbf{Escort}: Escort extends Surround by introducing a linearly moving target from an initial to a final position across a fixed time horizon. Agents must maintain a stable formation while tracking the moving target under the same dual objectives.

\subsection*{Appendix D: Implementation Details}

Our multi-agent multi-objective reinforcement learning framework combines MADDPG with attention mechanisms and preference learning for cooperative environments. The implementation consists of several key components detailed below.

\subsubsection*{Network Architecture:}

\textbf{Actor Network.} Each agent's actor network takes individual observations and preference vectors as input. The architecture follows: $\text{obs\_dim} \rightarrow 128 \rightarrow 256 \rightarrow \text{action\_dim}$ with ReLU activations, LayerNorm, and Tanh output activation for continuous actions.

\textbf{Critic Network.} The critic network employs a modular design with three components:
\begin{itemize}
    \item \textit{Agent Embedding Layer}: Encodes global state, concatenated actions, and preferences into 128-dimensional embeddings. Input dimension: $\text{state\_dim} + \sum \text{action\_dims} + \sum \text{preference\_dims}$.
    \item \textit{Central Attention Layer}: Multi-head self-attention mechanism with 8 heads, embedding dimension 128, and feed-forward network $(128 \rightarrow 256 \rightarrow 128 \rightarrow 128)$ with LayerNorm and residual connections.
    \item \textit{Output Layer}: Produces Q-values for each reward dimension following $(128 \rightarrow 512 \rightarrow 256 \rightarrow \text{reward\_dim})$.
\end{itemize}

\subsubsection*{Training Configuration:}

Key hyperparameters include: batch size 128, buffer size $5 \times 10^5$, discount factor $\gamma = 0.99$, soft update rate $\tau = 0.005$, learning rates (actor: $5 \times 10^{-4}$, critic: $3 \times 10^{-4}$), and 32 candidate actions for GPI.
We use Pytorch to implement all the deep learning models on our NVIDIA GeForce RTX 5090.

\end{document}